\documentclass[10pt,journal]{IEEEtran}
\makeatletter
\def\ps@headings{%
\def\@oddhead{\mbox{}\scriptsize\rightmark \hfil \thepage}%
\def\@evenhead{\scriptsize\thepage \hfil \leftmark\mbox{}}%
\def\@oddfoot{}%
\def\@evenfoot{}}
\makeatother
\usepackage{amssymb,color,amsmath,clock}
\usepackage{paralist}
\usepackage{graphicx,epsfig}
\usepackage{multicol}
\usepackage{arydshln}
\usepackage{graphicx}
\usepackage[urlcolor=rltblue,colorlinks=true]{hyperref} %
\definecolor{rltblue}{rgb}{0,0,0.75}
\usepackage{cite,manfnt}
\newcommand{\F}{\mathbf{F}}

\newcommand{\C}{\mathcal{C}}
\newcommand{\N}{\mathcal{N}}

\newtheorem{theorem}{\textbf{Theorem}}
\newtheorem{lemma}[theorem]{\textbf{Lemma}}
\newtheorem{definition}[theorem]{\textbf{Definition}}

\newtheorem{example}{Example}

\newcommand{\nix}[1]{}

\begin{document}
\title{\LARGE  Protection Over Asymmetric Channels\\ S-MATE: Secure Multipath Adaptive Traffic Engineering}

\author{Salah A. Aly$^{\dag\ddag}$~~~~~~ Nirwan Ansari$^\dag$~~~~~~H. Vincent Poor$^\ddag$~~~~~~Anwar I. Walid$^\S$ \\
 $^\dag$Dept. of Electrical  Engineering, New Jersey Inst. of Tech., Newark, NJ 07102, USA\\
$^\ddag$ Dept. of Electrical Engineering,~ Princeton University, ~Princeton, ~NJ 08544, ~USA\\
$^\S$Bell Laboratories~ \& Alcatel-Lucent, ~Murray Hill,  ~NJ 07974, ~USA
\thanks{The material in this paper was presented in part at  ICC'10,  Cape Town, South Africa, May 23-27, 2010. Available: arXiv:1010.4858 }
}
\markboth{Submitted, December 2010 }
{Submitted, December 2010 }
 \maketitle
\begin{abstract}
There have been several approaches to the problem of provisioning  traffic engineering between core network  nodes  in Internet Service Provider (ISP) networks. Such approaches aim to minimize network delay, increase capacity, and enhance security services between two core (relay) network nodes, an ingress node and an egress node.  MATE (Multipath Adaptive Traffic Engineering) has been proposed for multipath adaptive traffic engineering between an ingress node (source) and an egress node (destination) to distribute the network flow among multiple disjoint paths. Its novel idea is to avoid network congestion and attacks that might exist in edge and node disjoint paths between two core network nodes.

This paper aims to develop an adaptive, robust, and reliable traffic engineering scheme to improve performance and reliability of communication networks. This scheme will also provision  Quality of Server (QoS) and protection of  traffic engineering to maximize network efficiency. Specifically,   S-MATE  (secure MATE) is proposed to  protect the network traffic between two core nodes (routers, switches, etc.) in a cloud network. S-MATE secures against a single link attack/failure by adding redundancy in one of the operational redundant paths between the sender and receiver nodes. It is also extended to secure against multiple attacked links. The proposed scheme can be  applied to secure core networks  such as optical and IP networks.
\end{abstract}
\begin{keywords}
MATE Protocol, Network Coding, Adaptive Traffic Engineering, Internet Protection and Security.
\end{keywords}
\section{Introduction}\label{sec:intro}

Several approaches have been proposed for adapting the traffic between core network nodes  in Internet Service Provider (ISP) networks~\cite{elwalid02,he06,kandula05}.  Elwalid \emph{et al.}~\cite{elwalid02}  proposed an algorithm for multipath adaptive traffic engineering between an ingress node (source) and an egress node (destination). Their novel idea is to avoid network congestion that might exist in disjoint paths between two core  network nodes. They suggested load balancing among paths based on measurement and analysis of path congestion by using Multi-Protocol Label Switching (MPLS). MPLS  is a widely adopted tool for facilitating traffic engineering unlike explicit routing protocols, which allow certain routing methodology from hop-to-hop in a network with multiple core devices. The major advantage of MATE is that it does not require scheduling, buffer management, or traffic priority in the nodes.

In this work, we propose a new scheme, Secure Multipath Adaptive Traffic Engineering (S-MATE), that aims to protect/secure multiple disjoint paths for network traffic. S-MATE enables reliable data delivery and provides protection against link and router failures. The main feature of S-MATE is that the protection is achieved without retransmitting the lost packets or resending the ACK/NACK messages at the receivers. The sender keeps sending its data at a regular rate once the key $k$-disjoint paths are established. In addition, the proposed scheme provisions load balancing, meaning that the redundant data is distributed fairly among the available provisioned disjoint paths. Furthermore, once a certain path experiences delay or high risk of failures, the proposed scheme is modified to provide quality of service (QoS) traffic engineering. The latter scheme is referred to as QoS-S-MATE.

\begin{figure}[t]
\begin{center}
  \includegraphics[scale=0.62]{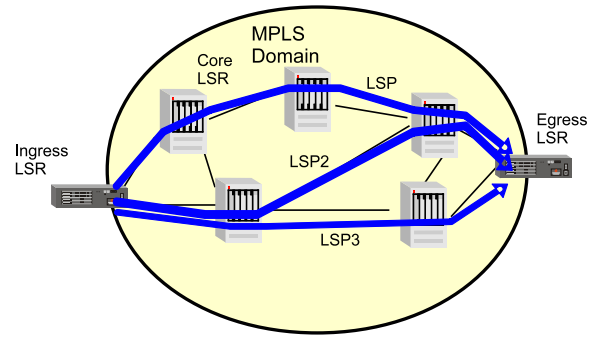}
  \caption{The network model is represented by two network nodes, ingress node (source) and egress node (receiver). There are $k$ link disjoint paths between the ingress and egress nodes. The link disjoint multipaths are established by using a network management software at the core routers.}
  \label{fig:netfig1}
  \end{center}
\end{figure}

Several recovery mechanisms against failures are proposed to ensure reliability and delivery of transmitted data by the core router nodes in the presence of link and relay failures~\cite{osborne02},\cite{sharma03},\cite{vasseur04}. These mechanisms also aim to guarantee  the Service Level Agreements (SLAs). Failures of links  and routers occur due to several reasons such as network component imperfections and changes of network topology. However, the protection operation is a challenging task because once the failure occurs the network traffic has to be rerouted among other routers, or delayed in the links for a short period of time. Such circumstances are unexpected and challenging for the network operators. One way to ensure data delivery is to establish backup paths between ingress and egress nodes.

Network coding is a powerful tool that has been recently used to increase
the throughput, capacity, and performance of wired and wireless communication networks.
Information theoretic aspects of network coding have been investigated
in References~\cite{ahlswede00},\cite{soljanin07},\cite{yeung06}, and in the list of references therein. It offers benefits in terms of
energy efficiency, additional security, and reduced delay. Network coding
allows the intermediate nodes not only to forward packets using network
scheduling algorithms, but also to encode/decode them through algebraic
primitive operations~\cite{ahlswede00,fragouli06,soljanin07,yeung06}. For example, data loss
because of failures in communication links can be detected and recovered
if the sources are allowed to perform network coding operations~\cite{cai06,gkantsidis06,jaggi07}.

Multipath Adaptive Traffic Engineering (MATE), which was previously proposed by one of the authors of this paper, is a traffic load balancing scheme that is suitable for S-MATE (secure MATE) as will be explained later.  MATE distributes traffic among the edge disjoint paths, so as to equalize the path delays. This is achieved by using adaptive algorithms. MATE has inspired other traffic engineering solutions such as TexCP~\cite{kandula05} and the measurement-based optimal routing solution~\cite{guven06}.
In this paper, we will design a security scheme by using network coding to protect against an entity who can not only copy/listen to the message, but also can fabricate new messages or modify the current ones. We aim to build an adaptive, robust, reliable traffic engineering scheme for better performance and operation of communication networks. The scheme will also provision  QoS and protection of traffic engineering to maximize network efficiency.

The rest of the paper is organized as follows. In Section~\ref{sec:model}, we present the network model and assumptions. In Sections~\ref{sec:MATE}, \ref{sec:NPSscheme} and~\ref{sec:SMATE}, we review the MATE algorithm and propose the secure MATE scheme based on network coding.  S-MATE against single and multiple attacks is presented in Sections~\ref{sec:twopaths}, \ref{sec:Tfilures}, \ref{sec:distributedcapacities}  and~\ref{sec:implementation}. Finally, Section~\ref{sec:conclusion} concludes the paper.

\section{Network Model and Assumptions}\label{sec:model}

The network model can be represented as follows. Assume a given network represented by a set of nodes and links.  The network nodes are core nodes that transmit outgoing packets to the neighboring nodes in certain time slots.  The network nodes are ingress and egress nodes that share multiple edge and node disjoint paths.

We assume that the core nodes share $k$ edge disjoint paths, as shown in Fig.~\ref{fig:netfig1}, for one particular pair of ingress and egress nodes. Let $N=\{N_1,N_2,...\}$  be the set of nodes (ingress and egress nodes) and $L=\{L_{\ell h}^1,L_{\ell h}^2,...,L_{\ell h}^k\}$ be the set of disjoint paths from an ingress node $N_{\ell}$ to an egress node $N_{h}$.  Every path $L_{\ell h}^i$ carries  segments of independent packets from an ingress node $N_\ell$ to egress node $N_h$. Let $P^{ij}_{\ell h}$ be the packet sent from the ingress node $N_\ell$ in path $i$ at time slot $j$ to the egress node $N_h$. For  simplicity, we describe the proposed scheme for one particular pair of ingress and egress nodes. Hence, we use $P^{ij}$ to represent a packet in path $i$ at time slot $j$.

Assume there are $\delta$ rounds (time slots) in a transmission session.  For the remainder of the paper, rounds and time slots will be used interchangeably. Packet $P^{ij}$ is indexed as follows:
\begin{eqnarray}\label{eq:plainpacket}
Packet_{\ell h}^{ij}(ID_{N_\ell}, X^{ij},round_j),
\end{eqnarray}
where $ID_{N_\ell}$ and $X^{ij}$ are the sender ID and transmitted data from $N_\ell$ in the path $L_i$ at time slot $j$.
There are two types of packets: plain and encoded packets. The plain packet contains the unencoded data from the ingress to egress nodes as shown in Equation~(\ref{eq:plainpacket}). The encoded packet contains encoded data from different incoming packets. For example, if there are $k$ incoming packets to the ingress node $N_l$, then the encoded data traversed in the protection path $L_{l h}^j$ to the egress node $N_h$ is given by

\begin{eqnarray}\label{eq:encodeddata}
y^j=\sum_{i=1, j\neq i}^k P_{l h}^{ij},
\end{eqnarray}
where  the summation denotes the binary addition. The corresponding packet becomes
\begin{eqnarray}\label{eq:encodedpacket}
Packet_{\ell h}^{ij}(ID_{N_\ell}, y^{j},round_j).
\end{eqnarray}

The following definition describes the \emph{working} and
\emph{protection} paths between two network switches as shown in
Fig.~\ref{fig:netfig1}.

\begin{definition}
The \emph{working paths} in a network with $n$ connection paths carry
un-encoded (plain) traffic under normal operations. The \emph{protection paths} provide
alternate backup paths to carry encoded traffic. A
protection scheme ensures that data sent from the sources will reach the
receivers in case of failures in the working paths.
\end{definition}

We make the following assumptions about the transmission of the plain and encoded packets.
\begin{compactenum}[i)]
\item The TCP protocol will handle the transmission and packet headers in the edge disjoint paths from the ingress to egress nodes.
\item The data from the ingress nodes  are sent in rounds and sessions throughout the edge disjoint paths to the egress nodes. Each session is quantified by the
    number of rounds (time slots) $n$. Hence, $t_j^\delta$ is the transmission time at the time
    slot $j$ in session $\delta$.

\item The attacks and failures on a path $L_i$ may be incurred by a network incident such as an eavesdropper, link replacement, and overhead. We
    assume that the receiver is able to detect a failure, and our
    protection strategy described in S-MATE is able to recover it.
\item We assume that the ingress and egress nodes share a set of $k$ symmetric     keys. Furthermore, the plain and encoded data are encrypted by using this set of keys. That is $$x^i=Encypt_{key_i} (m^i),$$ where $m_i$ is the message encrypted by the $key_i$. Sharing symmetric keys between two entities (two core network nodes) can be achieved by using key establishment protocols described in ~\cite{menezens01} and~\cite{schneier96}.
\item In this network model, we consider only a single link failure or attack; it is thus sufficient to apply the encoding and decoding operations over a finite
    field with two elements,  denoted as $\F_2=\{0,1\}$.
\end{compactenum}


 The traffic from the ingress node to the egress node in edge disjoint paths can be exposed to edge failures and network attacks. Hence, it is desirable to protect and secure this traffic. We assume that there is a set of $k$ connection paths that need to be fully guaranteed and protected
 against  a single  edge failure from ingress to egress nodes. We assume that all
connections have the same bandwidth, and each link (one hop or circuit) has the same bandwidth as the path.

The benefits of the proposed solutions include the following:
\begin{compactenum}[i)]
\item network protection is provisioned,
\item recovery is achieved without retransmitting the lost packets,
\item the sender can transmit  at a constant high rate,
\item the lost packets are recovered at the receiver online without sending an ACK message or notifying the sender about the failure, and
\item the network traffic is not rerouted or delayed.
\end{compactenum}

\section{MATE Protocol}\label{sec:MATE}
MPLS (Multipath Protocol Label Switching) is an emerging tool for facilitating network traffic  and out-of-band control.   Unlike explicit routing protocols, which allow certain routing methodology from hop-to-hop in a network with multiple core devices, MPLS balances network traffic.
As shown in Fig.~\ref{fig:netfig2}, MATE assumes that several explicit paths between an ingress node and an egress node in a cloud network  have been established. This is a typical setting which exists in operational Internet Service Providers (ISP) core networks (which implement MPLS).  The goal of the ingress node is to distribute traffic across the edge disjoint paths, so that the loads are balanced. One advantage of this load balancing is to equalize path delays, and to minimize traffic congestion~\cite{elwalid01,elwalid02}.


\begin{figure}[t]
\begin{center}
  \includegraphics[scale=0.6]{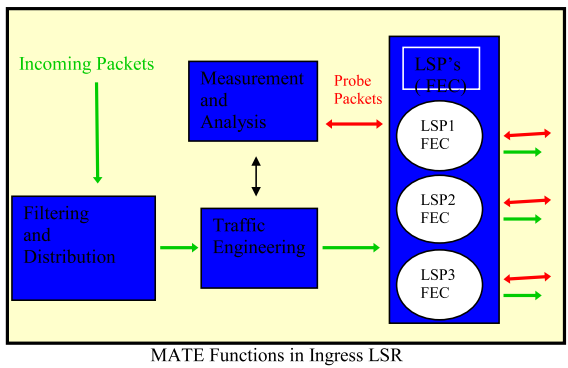}
  \caption{MATE traffic engineering at the ingress node.}
  \label{fig:netfig2}
  \end{center}
\end{figure}
The following are the key features of the MATE algorithm.
\begin{compactenum}[1)]
\item The traffic is distributed at the granularity of the IP flow level. This ensures that packets from the same flow follow the same path, and hence there is no need for packet re-sequencing at the destination.   This is easily and effectively  achieved by using a hashing function on the five tuple IP address.
\item
MATE  is a traffic load balancing scheme, which is suitable for S-MATE,  as will be explained later.  MATE distributes traffic among the edge disjoint paths, so as to equalize the paths delays. This is achieved by using adaptive algorithms as shown in Fig.~\ref{fig:netfig2} and Reference~\cite{elwalid02}

\item It is shown that the distributed load balancing (for each ingress,  egress pair) is stable and provably convergent. MATE assumes that several network nodes exist between ingress nodes as traffic senders and egress nodes as traffic receivers.  Furthermore, the traffic can be adapted by using switching protocols such as CR-LDP~\cite{rosen98} and RSVP-TE~\cite{bertsekas98}. An ingress node is responsible for managing the traffic in the multiple paths to the egress nodes so that  traffic congestion and overhead are minimized.
\end{compactenum}

As shown in Fig.~\ref{fig:netfig2}, Label Switch Paths (LSPs) from an ingress node to an egress node are provisioned before the actual packet is transmitted.  Then, once the transmissions start, the ingress node will estimate the congestion that might occur in one or more of the $k$ edge disjoint paths. As stated in Reference~\cite{elwalid02}, the congestion measure is related to one of the following factors: delay, loss rate, and bandwidth.  In general, each ingress node in the network will route the incoming packets into the $k$ disjoint paths. One of these paths will carry the encoded packets, and  all  other $k-1$ paths will carry plain packets. Each packet has its own routing number, so that the egress node will be able to manage the order of the incoming packet, and thus achieve the decoding operations.

As explained in~\cite{elwalid02}, MATE  works in two phases: a monitoring phase and a load balancing phase.  These two phases will monitor the traffic and balance packets among all disjoint paths. One beneficial feature of MATE is that its load balancing algorithms equalize the derivative of delay among all edge disjoint paths from an ingress node to an egress node. Furthermore, MATE's load balancing  preserves packet ordering since load balancing is done at the flow level (which is identified by a 5-tuple IP address) rather than at the packet level

We ensure that the proposed protocol in the following section is suitable for Internet traffic such as voice over IP (VoIP), multimedia teleconferencing, online gaming, TV streams. Such traffic is delay-sensitive and intolerant to late packet arrivals. This approach is different from other techniques for delay-sensitive traffic, including shortest path routing, or equal load-balancing splitting among multiple paths. As shown in a Cisco manuscript~\cite{cisco07-2012}, by 2012  video traffic will occupy $90\%$ of the total Internet traffic. Hence, techniques for delay minimization and online protection against failures are needed. Techniques that depend on shortest paths between ingress and egress nodes or on retransmitting  the lost packets appear to be impractical for delay sensitive traffics~\cite{suchara09}.

\section{Protection Using a
Dedicated Path}\label{sec:NPSscheme}
In this section, we present a Network Protection Strategy (NPS) against a single network failure. The single failure  could be one link or one core node (router or switch) in the given network topology.
Let $x_i^\ell$ be
the data sent from the source $s_i$ at round time $\ell$ in a session
$t_\delta^\ell$. Also, assume $y_j=\sum_{i=1,i \neq j}^k x_i^\ell$. Put
differently,
\begin{eqnarray}
y_{j}^\ell=x_1^\ell\oplus x_2^\ell\oplus \ldots  \oplus x_{i\neq j}^\ell  \oplus \ldots\oplus x_k^\ell.
\end{eqnarray}

The protection scheme  runs in sessions as explained below. Every session has at most
one single failure throughout  each round.

Some network topologies do not allow adding extra paths between the ingress and egress nodes. In this case, we propose that  one of the available working paths
can be used to carry the encoded data as shown in
~(\ref{scheme:k-1protection}). It  shows that  there exists a path $L_j$
that  carries the encoded data sent from the source $s_j$ to the receiver
$r_j$.

\begin{eqnarray}\label{scheme:k-1protection}
\begin{array}{|c|cccccc|c|c|}
\hline
\multicolumn{9}{|c|}{\mbox{ NPS scheme}}\\\multicolumn{9}{|c|}{\mbox{ }}\\
\hline
& \multicolumn{6}{|c|}{\mbox{ round time session 1 }}&\ldots&\ldots    \\
\hline
&1&2&3&\ldots&\ldots&n&\!\!\ldots&\ldots\\
\hline    \hline
  L_1 & x_1^1&x_1^2 &x_1^3&\ldots &\ldots  &x_1^n &\ldots&\ldots \\
    L_2 &  x_2^1&  x_2^2&x_2^3&\ldots&\ldots&x_2^{n} &\ldots&\ldots \\
L_3&  x_3^1& x_3^2& x_3^3&\ldots&\ldots&x_3^{n} &\ldots&\ldots \\
 \vdots&\vdots&\vdots&\vdots&\vdots&\vdots&\vdots&\ldots&\ldots\\
    L_i& x_i^1 & x_i^2 &\ldots&x_i^{i-1} &\ldots& x_i^{n}&\ldots&\ldots\\
     \vdots&\vdots&\vdots&\vdots&\vdots&\vdots&\vdots&\ldots&\ldots\\
     L_j & y_j^1&y_j^2&y_j^3&\ldots&\ldots&y_j^{n}&\ldots&\ldots\\
 \vdots&\vdots&\vdots&\vdots&\vdots&\vdots&\vdots&\ldots&\ldots\\
   L_k & x_k^1&x_k^2&x_k^3&\ldots&\ldots&x_{k}^n&\ldots&\ldots\\
\hline
\hline
\end{array}
\end{eqnarray}
All $y_j^\ell$'s are defined over $\F_2$ as \begin{eqnarray}
y_j^\ell=\sum_{i=1,i\neq j}^k x_i^\ell.
\end{eqnarray}

Note that the encoded data $y_j^\ell$ is fixed per one session
transmission but it is varied for other sessions. This means that the path
$L_j$ is dedicated to sending all encoded data $y_j^1,y_j^2,\ldots,y_j^n$, for all $1\leq j \leq k$. The normalized capacity of this scheme is still $(n-1)/n$.

\begin{lemma}
The normalized capacity of \textbf{NPS}  described
in~(\ref{scheme:k-1protection}) is given by
\begin{eqnarray}
\C=(k-1)/(k),
\end{eqnarray}
where $k$ is the number of disjoint paths.
\end{lemma}
\begin{proof}
We have $n$ rounds and the total number of transmitted packets in every
round is $k$. Also, in every round there are $(k-1)$ un-encoded data
$x_1,x_2,\ldots x_{i\neq j},\ldots,x_{k}$ and only one encoded data $y_j$,
for all $i=1,\ldots,n$. Hence, the capacity $c_\ell$ in every round is
$k-1$. Therefore, the normalized capacity is given by
\begin{eqnarray}
\C=\frac{\sum_{\ell=1}^n c_\ell}{k*n}=\frac{(k-1)*n}{kn}.
\end{eqnarray}
\end{proof}

 The following lemma shows that the network protection strategy
{\bf NPS} is in fact optimal if we consider the field $\F_2$. In other words, there
exist no other strategies that give better normalized capacity than
{\bf NPS}.

\begin{lemma}
The network protection shown in~(\ref{scheme:k-1protection}) against a single link failure is optimal.
\end{lemma}

The transmission is done in rounds, and hence  linear combinations of data
have to be from the same round. This can be achieved by using the round
time that is included in each packet sent by a sender.

\noindent \textbf{Encoding Process:} There are several scenarios in which the
encoding operations can be achieved. The encoding and decoding operations
will depend mainly on the network topology; how the senders and receivers
are distributed in the network.
 The encoding operation is done at only one source $s_i$ (ingress router). In this
    case, all other sources must send their data to $s_i$,  which will
    send encoded data over $L_i$. We assume that all sources share paths with each other.
\begin{figure}[t]
\begin{eqnarray}\label{eq:secMATE}
\begin{array}{|c|cccccc|c|}
\hline
& \multicolumn{6}{|c|}{\mbox{rounds from ingress to egress nodes}}&\ldots  \\
\hline
&1&2&3&\ldots&\ldots&n&\!\!\ldots\\
\hline    \hline
  L_{lh}^1& y^1&P^{11}&P^{12}&\ldots&\ldots&P^{1(n-1)} &\ldots   \\
    L_{lh}^2 &  P^{21}& y^2& P^{22}&\ldots&\ldots&\!\! P^{2(n-1)} &\ldots \\
L_{lh}^3 &  P^{31}& P^{32}&y^3& \ldots&\ldots&\!\! P^{3(n-1)} &\ldots \\
     \vdots&\vdots&\vdots&\vdots&\vdots&\vdots&\vdots&\ldots\\
     L_{lh}^j & P^{j1}&P^{j2}&\ldots&y^j&\ldots&\!\!P^{j(n-1)}&\ldots\\
 \vdots&\vdots&\vdots&\vdots&\vdots&\vdots&\vdots&\ldots\\
   L_{lh}^k & P^{k1}&P^{k2}&\ldots&\ldots&\!\!P^{k(k-1)}&y^n&\ldots\\
\hline
\hline
\end{array}
\end{eqnarray}
\end{figure}
\section{S-MATE}\label{sec:SMATE}
We assume that the network management software at the router level will compute the available disjoint paths between ingress and egress routers given the traffic demands, network flow, and capacity of communication links. In addition, it determines the network topology, failure locations, and failure causes. The proposed protocols will minimize  congestion  in the network operation in the presence of failures. We can also use one of the methods proposed in~\cite{suchara09} to compute the available multiple disjoint paths and be aware of the routers' conditions.

Traffic splitting in MPLS is deployed in today's routers~\cite{osborne02}.  This is also done in a flexible way such that packets belonging to the same traffic or coming from the same IP source will travel throughout the same path. Also, the path failure detection can be done using  detection protocol such as Bidirectional Forwarding Detection (BFD)~\cite{katz09}. As explained in~\cite{suchara09}, BFD establishes connections between two routers, ingress and egress nodes, to monitor  the traffic paths.

We now propose a scheme for securing MATE, called S-MATE (Secure Multipath Adaptive Traffic Engineering). The basic idea of S-MATE can be described  by Equation~(\ref{eq:secMATE}). S-MATE inherits the traffic engineering components described in the previous section and in References~\cite{elwalid01} and~\cite{elwalid02}.

Without loss of generality, assume that the network traffic between a pair of ingress and egress nodes is transmitted in $k$
edge disjoint paths, each of which carries  different packets. The disjoint paths are already established between the core nodes using any provisioning mechanism. Our proposed solution will protect these disjoint paths in case a failure occurs in one (or more) particular link(s) throughout one (or more) paths.

 The transmission of ingress (source) packets is achieved in rounds. For simplicity, we assume that the number of edge disjoint paths and the number of rounds in one transmission session are equal. Otherwise, the total number of rounds can be divided into $k$ separate rounds.
There are  two types of packets:
\begin{compactenum}[i)]
\item {\bf Plain Packets:} Packets $P^{ij}$ sent without coding, in which the ingress node does not need
    to perform any coding operations. For example, in case of packets
    sent without coding, the ingress node  $N_l$ sends the following packet
    to the egress node  $N_h$:
\begin{eqnarray}
packet_{N_l \rightarrow N_h}(ID_{N_l},x^{ij},t_\delta^j),~for~ i=1,2,..,k, i\neq j.
\end{eqnarray}
The plain data $x^{ij}$ is actually the encryption of the message $m^{ij}$ obtained by using any secure symmetric encryption algorithm~\cite{menezens01}. That is, $x^{ij}=Encypt_{key_i} (m^{ij})$, where $key_i$ is a shared symmetric key between $N_l$ and $N_h$.

\item  {\bf Encoded Packets:} Packets $y^i$ sent with encoded data, in which the ingress node $N_l$
    sends other incoming data. In this case, the ingress node $N_l$ sends the following packet to egress node
$N_h$:
\begin{eqnarray}
packet_{N_l \rightarrow N_h}(ID_{N_l},
\sum_{i=1}^{j-1} x^{i~j-1}+\sum_{i=j+1}^k x^{ij},t^j_\delta).
\end{eqnarray}
The encoded packet will be used in case  any  of the working paths is compromised. The egress node will be able to detect the compromised data, and can recover it by using the data sent in the protection path.
\end{compactenum}

\begin{lemma}
The S-MATE scheme  is optimal against a single  link attack.
\end{lemma}

What we mean by optimal here is that the encoding and decoding operations are achieved over the binary field with the least computational overhead. That is, one cannot find a better scheme than this proposed encoding scheme in terms of encoding operations. Indeed, one single protection path is used in case of a single attack path or failure.
The transmission is done in rounds (time slots), and hence  linear combinations of data
must be from the same round time. This can be achieved by using the
time slot that is included in each packet sent by the ingress node.

\begin{lemma}
The network capacity between the ingress node and the egress node is given by $k-1$ in the case of one single attack path.
\end{lemma}

\subsection{\textbf{Encoding Process}} There are several scenarios in which the
encoding operations can be achieved. The encoding and decoding operations
will depend mainly on the network topology,  i.e., how the senders and receivers
are distributed in the network.
\begin{itemize}
\item The encoding operation is done at only one ingress node $N_l$. In this case,
     $N_l$ will prepare and send the encoded data over $L_{lh}^j$ to the receiver $N_h$.
\item We assume that $k$ packets will be sent  in every transmission session from the ingress node. Also, if the number of incoming packets is greater than $k$, then a modulo function is used to moderate the outgoing traffic in $k$ different packets. Each packet will be sent in  one unique path.
    \item Incoming packets with large sizes will be divided into small chunks of equal size.
\end{itemize}

\subsection{\textbf{Decoding Process}}
The decoding process is done in a similar way as explained in the previous work shown in~\cite{aly08preprint1} and \cite{aly09-4}.

We assume that the ingress node $N_l$  assigns
the paths that will carry plain data as shown in
Fig.~\ref{fig:netfig3}. In addition, $N_l$ will encode the
data from all incoming traffic and send it in one path.
This will be used to protect  any  single link  attacks/failure. The objective is to withhold rerouting the signals or the transmitted packets due to link attacks. However, we provide strategies that utilize network coding and reduced capacity at the ingress nodes. We assume that the source nodes (ingress) are able to perform encoding operations and the receiver nodes (ogress) are able to perform decoding operations.\\

One of S-MATE's objectives is to minimize the delay of the transmitted packets. So, the packets from one IP address will be received in order in one path. The following are the key features of S-MATE.
\begin{itemize}
\item The traffic from the ingress node to the egress node is secured against eavesdropper and intruders.

    \item No extra paths in addition to the existing network edge disjoint paths are needed to secure the network traffic.
        \item It can be implemented without adding new hardware or network components.

\end{itemize}

\begin{figure}[t]
\begin{center}
 \includegraphics[scale=0.70]{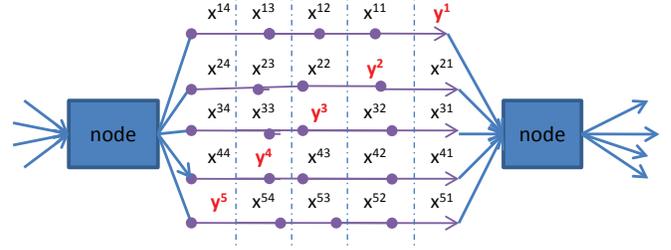} 
  \caption{Working and protection edge disjoint paths between two core nodes. The protection path carries encoded packets from all other working paths between ingress and egress nodes.}
  \label{fig:netfig3}
  \end{center}
\end{figure}

The following example illustrates the plain and encoded data transmitted from five senders to five receivers.

\medskip

\begin{example}
Let $N_l$ and $N_h$ be two core network nodes (a  sender and receiver) in a cloud network. Equation~(\ref{eq:example1}) explains the plain and encoded data sent in five consecutive time slots from the  sender to the  receiver. In the first time slot, the first connection carries encoded data, and all other connections carry plain data. Furthermore, the encoded data is distributed among  all connections in the time slots $2, 3, 4$ and $5$.
\begin{eqnarray}\label{eq:example1}
\begin{array}{|c|ccccc|c|c|}
\hline
cycle& &~~~~1&&&&2&3\\
\hline
rounds&1&2&3&4&5&\ldots&\ldots\\
\hline
\hline
L_{lh}^1 &y^1&x^{11}  &x^{12} &x^{13} &x^{14}& \ldots&\ldots\\
L_{lh}^2 &x^{21}&y^2  &x^{22} &x^{23}&x^{24}&\ldots&\ldots\\
L_{lh}^3 &x^{31}  &x^{32}&y^3 &x^{33}&x^{34}&\ldots&\ldots \\
L_{lh}^4    &x^{41} &x^{42} & x^{43}&y^{4}&x^{44}&\ldots&\ldots\\
L_{lh}^5    &x^{51} &x^{52} & x^{53}&x^{54}&y^5&\ldots&\ldots\\
\hline
\end{array}
\end{eqnarray}
The encoded data $y^j$, for $1 \leq j \leq 5$,  is sent as

\begin{eqnarray}
y^j=\sum_{i=1}^{j-1} x^{i~j-1}+\sum_{i=j+1}^5 x^{ij}.
\end{eqnarray}
We notice that every message has its own time slot. Hence, the protection data is distributed among all paths for fairness.
\end{example}


\section{A Strategy Against two attacked Paths}\label{sec:twopaths}
In this section, we propose a strategy against two attacked paths (links), i.e., securing MATE against two-path attacks. The strategy is achieved by using network coding and dedicated paths.  Assume we have $n$ connections carrying data from an ingress node to an egress node. All connections represent disjoint paths.

We will  provide two backup paths to secure against any two disjoint paths,
which might experience any sort of attacks. These two protection paths can be chosen by
using network provisioning. The protection paths are fixed for all rounds
per session from the ingress node to the egress node, but they may vary among sessions. For example, the ingress node $N_l$  transmits a message $x^{i\ell}$ to  the egress node $N_h$ through path $L_{\ell h}^i$ at time $t_\delta^\ell$ in round time $\ell$ in session $\delta$. This process is explained
in Equation~(\ref{eq:n-1protection2}) as follows:

\begin{eqnarray}\label{eq:n-1protection2}
\begin{array}{|c|ccccc|c|}
\hline
& \multicolumn{5}{|c|}{\mbox{ cycle 1 }}&\ldots    \\
\hline
&1&2&3&\ldots&n&\!\!\ldots\\
\hline    \hline
 L_{lh}^1 & x^{11}&x^{12} &x^{13}&\ldots  &x^{1n} &\ldots \\
   L_{lh}^2 &  x^{21}&  x^{22}&x^{23}&\ldots&x^{2n} &\ldots \\
L_{lh}^3 &  x^{31}& x^{32}& x^{33}&\ldots&x^{3n} &\ldots \\
 \vdots&\vdots&\vdots&&\vdots&\vdots&\ldots\\
L_{lh}^i  & x^{i1} & x^{i2} &x^{i3}&\ldots& x^{in}&\ldots\\
     L_{lh}^j & y^{j1}&y^{j2}&y^{j3}&\ldots&y^{jn}&\ldots\\
    L_{lh}^k  & y^{k1}&y^{k2}&y^{k3}&\ldots&y^{kn}&\ldots\\
     L_{lh}^{i+1}& x^{(i+1)1} & x^{(i+1)2} &x^{(i+1)3}&\ldots& x^{(i+1)n}&\ldots\\
 \vdots&\vdots&\vdots&\vdots&\vdots&\vdots&\ldots\\
  L_{lh}^n & x^{n1}&x^{n2}&x^{n3}&\ldots&x^{nn}&\ldots\\
\hline
\hline
\end{array}
\end{eqnarray}
All $y_j^\ell$'s are defined as
\begin{eqnarray} y^{j\ell}=\sum_{i=1,i\neq j \neq k}^n a_i^\ell x^{i\ell} \mbox{  and  } y^{k\ell}=\sum_{i=1,i\neq k \neq j}^n b_i^\ell x^{i\ell}.
\end{eqnarray}

The coefficients $a_i^\ell$ and $b_i^\ell$ are chosen over a finite field
$\F_q$ with $q > n-2$; see~\cite{aly08preprint1,Ayanoglu93} for more details. One way to choose these coefficients is by using the following two vectors:
\begin{eqnarray}\label{eq:twovectors}
\left[\begin{array}{ccccc}
1&1&1&\ldots&1\\
1&\alpha&\alpha^2&\ldots&\alpha^{n-3}
\end{array}\right].
\end{eqnarray}
Therefore, the coded data is
\begin{eqnarray} y^{j\ell}=\sum_{i=1,i\neq j \neq k}^n  x^{i\ell} \mbox{  and  } y^{k\ell}=\sum_{i=1,i\neq k \neq j}^n \alpha^{i \mod n-2} x^{i\ell}.
\end{eqnarray}
In the case of two failures, the receivers will be able to solve two linearly independent equations with two unknown variables. For instance, assume the two failures occur in paths number two and four. Then, the receivers will be able to construct two equations with coefficients
\begin{eqnarray}\label{eq:twovectors2}
\left[\begin{array}{cc}
1&1 \nonumber\\
\alpha&\alpha^{3}
\end{array}\right].
\end{eqnarray}
 Therefore, we have
\begin{eqnarray}
x^{2\ell}+x^{4\ell}\\
\alpha x^{2\ell}+\alpha^3 x^{4\ell}.
\end{eqnarray}
One can multiply the first equation by $\alpha$ and subtract the two equations to obtain the value of $x^{4\ell}$.

Note that the encoded data symbols
$y^{j\ell}$ and $y^{k\ell}$ are fixed for one session, but they are
varied for other sessions. This means that the path $L_{lh}^j$ is dedicated to
send all encoded data $y^{j1},y^{j2},\ldots,y^{jn}$.
\begin{lemma}
The  network capacity of the protection strategy against two-path attacks is given by
$n-2$.
\end{lemma}

There are three different scenarios for two-path attacks, which
can be described as follows:
\begin{compactenum}[i)]
\item If the two-path attacks occur in the backup protection paths $L_{lh}^j$
    and
    $L_{lh}^k$,
    then no recovery operations are required at the egress node.
\item If the two-path attacks  occur in one backup protection path, say
    $L_{lh}^j$, and one working path $L_{lh}^i$, then  recovery operations are
    required.
\item If the two-path attacks occur in two working paths, then in this
    case the two protection paths are used to recover the lost data. The
    idea of recovery in this case is to build a system of two linearly independent equations
    with two unknown variables.
\end{compactenum}

\section{Multiple Protection Paths Using S-MATE}
\label{sec:Tfilures}
In this section, we present S-MATE against $t$
attacked paths. We adopt  the same notations as in the
previous sections.  Assume also that the total number of attacks is $t$, and they
happen on arbitrary $t$ paths from the ingress node to the egress node.

\begin{figure*}[t]
\begin{center}
\begin{eqnarray}\label{eq:tFailuresScheme}
\begin{array}{|c|cccccc|}
\hline
\multicolumn{7}{|c|}{\mbox{ NPS-T Scheme}}
\\
\hline
&1&2&\ldots&j&\ldots&m=\lceil n/t\rceil\\
\hline    \hline
s_1 \rightarrow r_1 & y_1&x_1^1 &\ldots&x_1^{j-1}&\ldots &x_{1}^{m-1} \\
s_2 \rightarrow r_2 &  y_2&  x_2^1 &\ldots&x_2^{j-1}&\ldots&x_2^{m-1}  \\
\vdots&\vdots&\vdots&\vdots&\vdots&\vdots&\vdots\\
s_{t} \rightarrow r_{t} &  y_{t}&  x_{t}^1 &\ldots&x_{t}^{j-1}&\ldots&x_{t}^{m-1}  \\
\!\! s_{t+1} \! \rightarrow \!\! r_{t+1} \!\! &\!\! x_{t+1}^1\!&\! y_{t+1}&\ldots&x_{2t+1}^3&\! \ldots&\! x_{2t+1}^{m-1} \!\! \\
\vdots&\vdots&\vdots&\vdots&\vdots&\vdots&\vdots\\
s_{2t} \rightarrow r_{2t}&x_{2t}^1&y_{2t}&\ldots& x_{2t}^3&\ldots&x_{2t}^{m-1}  \\
\vdots\ddots&\vdots\ddots&\vdots\ddots&\vdots\ddots&\vdots\ddots&\vdots\ddots&\vdots\ddots\\
\!\! s_{jt+\ell}\!\! \rightarrow \!\!r_{jt+\ell}\!\!&\!\!x_{jt+\ell}^1\!&\!x_{jt+\ell}^2&\ldots&\! y_{jt+\ell}^3& \ldots&\! x_{jt+\ell}^{m-1} \!\! \\
\vdots\ddots&\vdots\ddots&\vdots\ddots&\vdots\ddots&\vdots\ddots&\vdots\ddots&\vdots\ddots\\
\!\! s_{t(m-1)+1} \rightarrow r_{t(m-1)+1}&x_{t(m-1)+1}^1&x_{t(m-1)+1}^2& \ldots&x_{t(m-1)+1}^j&\! \ldots&\! y_{t(m-1)+1}  \\
\vdots&\vdots&\vdots&\vdots&\vdots&\vdots&\vdots\\
s_{mt} \rightarrow r_{mt} & x_{mt}^1&x_{mt}^2&\ldots&x_{mt}^j&\ldots&y_{mt}\\
\vdots\ddots&\vdots\ddots&\vdots\ddots&\vdots\ddots&\vdots\ddots&\vdots\ddots&\vdots\ddots\\
\hline
\end{array}
\end{eqnarray}
\caption{The encoding scheme of $t$ link failures. $m=\lceil n/t\rceil$, $1 \leq j \leq m$ and $1 \leq \ell \leq t$. $t$ out of the $n$ connections carry encoded data. The coefficients are chosen over $\F_q$, for $q \geq n-t+1$.}
\end{center}
\end{figure*}

Let $m=\lceil n/t\rceil$, and hence we have $m$ rounds per cycle. The encoding operations of NPS-T against $t$ attacks/failures are described by~(\ref{eq:tFailuresScheme}).
We can see that $y_\ell$ in general  is given by

\begin{eqnarray}\label{eq:y_NPS-T}
y_\ell=\sum_{i=1}^{(j-1)t} a_i^\ell x_i^{j-1} + \sum_{i=jt+1}^n a_i^\ell x_i^j  \nonumber \\ \mbox{   for  } (j-1)t+1 \leq \ell \leq jt,  1 \leq j \leq n.
\end{eqnarray}


\subsection{Encoding Operations}
Assume that each connection path
$L_i$  has a unit capacity from an ingress source  $s_i$  to an
egress receiver $r_i$. The data sent from the source $s_i$ to the receiver $r_i$
is transmitted in rounds.
Under NPS-T, in every round $n-t$ paths are used to carry new data
($x_i^j$), and $t$ paths are used to carry protected data units.
There are $t$ protection paths.
Therefore, to treat all connections fairly,
there will be $n/t$ rounds in a cycle, and in each round the
capacity is given by $n-t$ from the ingress node to the egress node.

We consider the case in which all symbols $x_i^j$  belong to the same round.
The first $t$ sources transmit the first encoded data units
$y_1,y_2,\ldots,y_{t}$, and
in the second round, the next $t$ sources
transmit $y_{t+1},y_{t+2},\ldots,y_{2t}$, and so
on. The ingress and egress nodes  must keep
track of the round numbers. Let $ID_{s_i}$ and $x_{s_i}$ be the ID and
data initiated by the source $s_i$. Assume the round time $j$ in cycle
$\delta$ is given by $t^{j}_{\delta}$. Then, the source $s_i$ will send
$packet_{s_i}$ on the working path which includes
\begin{eqnarray}
Packet_{s_i}=(ID_{s_i}, x_{i}^\ell, t^\ell_\delta).
\end{eqnarray}
Also, the source $s_j$, which transmits on a protection path, will
send a packet $packet_{s_j}$:
\begin{eqnarray}
Packet_{s_j}=(ID_{s_j}, y_j, t^\ell_\delta),
\end{eqnarray}
where $y_k$ is defined as

\begin{eqnarray}\label{eq:y_NPS-T}
y_\ell=\sum_{i=1}^{(j-1)t} a_i^\ell x_i^{j-1} + \sum_{i=jt+1}^n a_i^\ell x_i^j  \nonumber \\ \mbox{   for  } (j-1)t+1 \leq \ell \leq jt,  1 \leq j \leq n.
\end{eqnarray} Hence, the protection
paths are used to protect the data transmitted in round $\ell$, which
are included in the $x^l_i$ data units.
 So, we have a system of $t$ independent equations at each round time
 that will be  used to recover at most $t$ unknown variables.

The strategy NPS-T is a generalization of protecting against a single path failure shown in the previous section in which $t$
protection paths are used instead of one protection path in case of
one failure.

\begin{theorem}
Let $n$ be the total number of connections from the ingress node to the egress node. The
capacity of NPC defined over $\F_q$ against $t$ path attacks is given by
\begin{eqnarray}
\C_{\N}=(n-t)/(n)
\end{eqnarray}
\end{theorem}

\subsection{Proper Coefficients Selection}
One way to select the coefficients $a_j^\ell$ in each round such that we
have a system of $t$ linearly independent equations is by using the matrix $H$ shown in Eq.~(\ref{bch:parity}). Let $q$ be the order of a finite field, and
$\alpha$ be the $q^{th}$ root of unity. Then, we can use this matrix   to define the coefficients of the senders as:
\begin{eqnarray}\label{bch:parity} H =\left[
\begin{array}{ccccc}1&1&1&\ldots&1\\1 &\alpha &\alpha^2 &\cdots &\alpha^{n-1}\\1
&\alpha^2 &\alpha^4 &\cdots &\alpha^{2(n-1)}\\\vdots& \vdots &\vdots
&\ddots &\vdots\\1 &\alpha^{t-1} &\alpha^{2(t-1)} &\cdots
&\alpha^{(t-1)(n-1)}\end{array}\right].\end{eqnarray}
We make the following assumptions about the encoding operations.
\begin{compactenum}
\item  Clearly, if we have one failure $t=1$, then all coefficients will be
    one. The first sender will always choose the unit value.

\item  If we assume $t$ failures, then  $y_1,y_2,\ldots,y_t$
    are written as:
\begin{eqnarray}
y_1&=&\sum_{i=t+1}^nx_i^1, ~~~~~
y_2=\sum_{i=t+1}^n \alpha^{(i-1) }x_i^2,
\\
\label{eq:tcofficients}
y_j&=&\sum_{i=t+1}^n \alpha^{i(j-1) \mod (q-1)}x_i^\ell.
\end{eqnarray}
\end{compactenum}

The previous  equation gives the general theme to choose the coefficients at any particular round in any cycle. However, the encoded data $y_i$'s are defined as shown in~(\ref{eq:tcofficients}). In other words, for the first round in cycle one, the coefficients of the plain data $x_1,x_2,\ldots,x_t$ are set to zero.

\section{Network Protection Using Distributed Capacities and QoS}\label{sec:distributedcapacities}

In this section, we develop a network protection strategy in which some connection paths (network traffic) have high priorities (less bandwidth and high demand). Let $k$ be  the set of available connections (disjoint paths from ingress to egress nodes carrying network traffic). Let $m$ be the set of rounds in every  cycle. We assume that all connection paths might not have the same priority demand and working capacities. The assigned priority itself can be done by using management software. This  can also be achieved by looking at the packet headers and checking what kind of traffic they carry. Also, the priority can  depend on the source IP address.  Connections that carry applications with multimedia traffic have higher priorities than those of applications carrying data traffic. Therefore, it is required to design network protection strategies based on the traffic and sender priorities.

Consider that available working connections $k$ may use their bandwidth assignments in asymmetric ways. Some connections are less demanding in terms of bandwidth requirements than other connections that require full capacity frequently. Therefore, connections with less demand can transmit more protection packets, while other connections demand more bandwidth, and can therefore transmit fewer protection packets throughout transmission rounds. Let $m$ be the number of rounds and $t_i^\delta$ be the time of transmission in a cycle $\delta$ at round $i$. For a particular cycle $i$, let $t$ be the number of protection paths against $t$ link failures or  attacks that might affect the working paths. We will design a network protection strategy against $t$ arbitrary link failures  as follows. Let the source $s_j$ send $d_i$ data packets and $p_i$ protection packets such that $d_j+p_j=m$. That is,

\begin{eqnarray}
\sum_{i=1}^k (d_i+p_i)=km.
\end{eqnarray}
In general, we do not assume that $d_i =d_j$ and $p_i=p_j$.

\begin{figure}
\begin{eqnarray}\label{eq:tfailures2}
\begin{array}{|c|ccccccc|}
\hline\multicolumn{8}{|c|}{\mbox{ QoS S-MATE Scheme}}
\\
\hline
& \multicolumn{7}{|c|}{\mbox{ round time cycle 1 }} \\ \hline
\hline
&1&2&3&4&\ldots&m-1&m\\
\hline    \hline
L_{\ell h}^1 & y^{11}&x^{11} &x^{12}&y^{12}&\ldots &y^{1p_1}&x^{1d_1} \\
L_{\ell h}^2 &  x^{21}& y^{21} &x^{22}&x^{23}&\ldots&x^{2d_2} &y^{2p_2} \\
\vdots&\vdots&\vdots&\vdots&\vdots&\vdots&\vdots&\vdots\\
L_{\ell h}^{i} &  y^{i1}&  x^{i1} &x^{i2}&y^{i2}&\ldots&y^{ip_i}  &x^{id_i}\\
\vdots&\vdots&\vdots&\vdots&\vdots&\vdots&\vdots&\vdots\\
L_{\ell h}^{j} &  x^{j1}&  x^{j2} &y^{j1}&x^{j3}&\ldots&x^{jd_j}&y^{jp_j}  \\
\vdots&\vdots&\vdots&\vdots&\vdots&\vdots&\vdots&\vdots\\
L_{\ell h}^k & x^{k1}&y^{k1}&x^{k2}&x^{k4}&\ldots&y^{kp_k}&x^{kd_k}\\
\hline
\end{array}
\end{eqnarray}
\end{figure}
The encoded data $y^{i\ell}$ is given by
\begin{eqnarray}
y^{i\ell} =\sum_{k=1,y^{k\ell} \neq y^{k\ell}} x^{k\ell}.
\end{eqnarray}

We assume that the maximum number of attacks/failures that might occur in a particular cycle is $t$. Hence, the number of protection paths (paths that carry encoded data) is $t$. The selection of the working and protection paths in every round is done by using a demand-based priority function at the senders's side. It will also depend on the traffic type and service provided on these protection and working connections. See Fig.~\ref{fig:netfig4} for  ingress and egress nodes with five disjoint connections.

In Eq.~(\ref{eq:tfailures2}), every connection $i$ is used to carry $d_i$ unencoded data $x^{i1},x^{i2},\ldots,x^{id_i}$ (working paths) and $p_i$ encoded data $y^{i1},y^{i2},\ldots,y^{ip_i}$ (protection paths) such that $d_i+p_i=m$.

\begin{lemma}\label{lem:nps-t2}
Let $t$ be the number of connection paths carrying encoded data in every round. The  network capacity $C_\N$ is then given by
\begin{eqnarray}
C_\N=k-t.
\end{eqnarray}
\end{lemma}
\begin{IEEEproof}
The proof is forward straight from the fact that $t$ protection paths exist in every round among the $k$ available disjoint paths, and hence $k-t$ working paths are available throughout all $m$ rounds.
\end{IEEEproof}

\section{Practical Aspects}\label{sec:implementation}
The network protection strategy against
a link failure is deployed in two processes: encoding and decoding
operations. The  encoding operations are performed at the ingress router,
which will send the encoded data depending on the
 adapted strategy throughout the available multipaths. The packets are sent in rounds. Each packet is marked by using the current round time and the path number. This is achieved till all packets are sent throughout all paths.

The decoding operations are performed at the receiver
side (egress router), which will apply  XOR operations to all incoming traffic to recover the lost packets in case of a single link failure.  If the receivers can tolerate a large amount of delay as in the
case of storage files, then,  the S-MATE strategy can be used. For applications
that cannot tolerate packet delays (delay sensitive traffic) such as multimedia or TV streams, the
S-MATE strategy can be used. We also note that the delay will occur only
when a failure occurs in the protection paths.

The transmission is done in rounds, and hence  linear combinations of data have
to be from the same round. This can be achieved by using the round time
that is included in each packet sent by a sender.

The core routers will manage the available multipaths by using network management software.  In this case, the number of link disjoint paths are known and provisioned in advance. Furthermore, the routers will decide which protection strategies will be used depending on the network conditions and number of failures.
\bigskip


\section{Conclusion}\label{sec:conclusion}
In this paper, we have proposed the S-MATE scheme (secure multipath adaptive traffic engineering) for operational networks. We have utilized network coding of transmitted packets to protect the traffic between two network core nodes (routers, switches, etc.) that could exist in a cloud network. Our assumption is based on the fact that core network nodes share multiple edge disjoint paths. S-MATE  can secure network traffic against  single link attacks/failures by adding redundancy in one of the operational paths between the sender and receiver. It can also be used to secure network traffic against two and multiple attacks/failures. The proposed scheme can be  built to secure operational networks  including optical and multipath adaptive networks. In particular, it can provide security services at the IP and data link layers.

\begin{figure}[t]
\begin{center}
  \includegraphics[width=9cm,height=4cm]{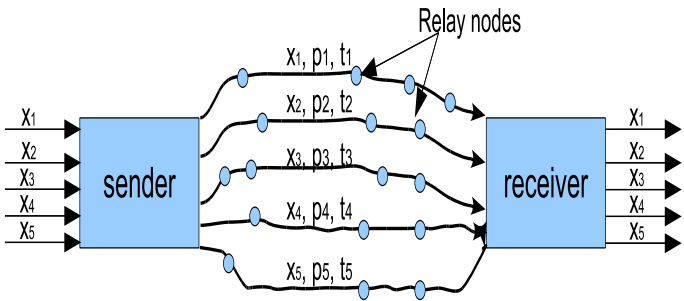}
  \caption{Working and protection edge disjoint paths between two core nodes (ingress and egress nodes). Every path $L_i$ carries encoded and plain packets depending on the traffic priority $p_i$.}
  \label{fig:netfig4}
  \end{center}
\end{figure}

\bigskip

\bibliographystyle{plain}

\bibliographystyle{ieeetr}
\end{document}